%% file: icdcs21.tex
\documentclass[conference]{IEEEtran}
\IEEEoverridecommandlockouts
\usepackage{cite}
\usepackage{amsmath,amssymb,amsfonts}
\usepackage{graphicx,comment}
\usepackage{textcomp}
\usepackage{xcolor}
\usepackage{cuted}
\def\BibTeX{{\rm B\kern-.05em{\sc i\kern-.025em b}\kern-.08em
    T\kern-.1667em\lower.7ex\hbox{E}\kern-.125emX}}

\usepackage{bbold}
\usepackage{graphicx}
\usepackage{gensymb}
\usepackage{amsmath}
\usepackage{amssymb}
\usepackage{amsthm}
\usepackage{epstopdf}
\usepackage{color}
\usepackage{hyperref}
\usepackage{subcaption,bbm}
\usepackage{multirow}
\usepackage{algorithm}
\usepackage[noend]{algpseudocode}
\newtheorem{theorem}{Theorem}
\newtheorem{lemma}{Lemma}

\newtheorem{definition}{Definition}
\newtheorem{remark}{Remark}
\newtheorem{proposition}{Proposition}

\newcommand{\bE}{\mathbb{E}}
\newcommand{\bP}{\mathbb{P}}
\newcommand{\te}{\theta}
\newcommand{\lmu}{_{i,\mu}}
\newcommand{\llm}{_{i,\lambda}}
\newcommand{\cE}{\mathcal{E}}

\newcommand{\up}{^{\prime}}
\newcommand{\cO}{\mathcal{O}}

\newcommand{\cG}{\mathcal{G}}
\newcommand{\id}{\mathbbm{1}}

\newcommand{\bN}{\mathbb{N}}
\newcommand{\cL}{\mathcal{L}}

\newtheorem{assumption}{Assumption}

\begin{document}

\title{A Constrained Reinforcement Learning Based Approach for Service Placement over Wireless Edge}

\title{Resource Constrained Reinforcement Learning for Service Placement over the Wireless Edge}

\title{Learning Augmented Whittle Index Policy for Service Placement at the Network Edge}

\title{Learning Augmented Index Policy for \\Optimal Service Placement at the Network Edge}

\author{\IEEEauthorblockN{Guojun Xiong$^*$\thanks{$^*$Equal contribution.  Ordering determined by inverse alphabetical order.}}
\IEEEauthorblockA{\textit{SUNY-Binghamton University} \\
Binghamton, NY 13902, USA \\
gxiong1@binghamton.edu}
\and
\IEEEauthorblockN{Rahul Singh$^*$}
\IEEEauthorblockA{\textit{Indian Institute of Science} \\
Bengaluru, Karnataka 560012, India \\
rahulsingh@iisc.ac.in}
\and
\IEEEauthorblockN{Jian Li$^*$}
\IEEEauthorblockA{\textit{SUNY-Binghamton University} \\
Binghamton, NY 13902, USA \\
lij@binghamton.edu}
}

\maketitle

\input{abstract}

\begin{IEEEkeywords}
Service placement, network edge, Markov decision process, Whittle index policy, reinforcement learning
\end{IEEEkeywords}

\input{intro}

\input{model}

\input{relaxation}

\input{learning}

\input{proofucb}

\input{simulation}

\input{related}

\input{conclusion}

\bibliographystyle{IEEEtran}
\bibliography{refs}

\clearpage
\input{appendixnonlearning}

\input{appendix}

\end{document}

%% file: abstract.tex
\begin{abstract}
We consider the problem of service placement at the network edge, in which a decision maker has to choose between $N$ services to host at the edge to satisfy the demands of customers. Our goal is to design adaptive algorithms to minimize the average service delivery latency for customers. We pose the problem as a Markov decision process (MDP) in which the system state is given by describing, for each service, the number of customers that are currently waiting at the edge to obtain the service.  However, solving this $N$-services MDP is computationally expensive due to the curse of dimensionality. To overcome this challenge, we show that the optimal policy for a single-service MDP has an appealing threshold structure, and derive explicitly the Whittle indices for each service as a function of the number of requests from customers based on the theory of Whittle index policy.

Since request arrival and service delivery rates are usually unknown and possibly time-varying, we then develop efficient learning augmented algorithms that fully utilize the structure of optimal policies with a low learning regret.  The first of these is UCB-Whittle, and relies upon the principle of optimism in the face of uncertainty. The second algorithm, Q-learning-Whittle, utilizes Q-learning iterations for each service by using a two time scale stochastic approximation. We characterize the non-asymptotic performance of UCB-Whittle by analyzing its learning regret, and also analyze the convergence properties of Q-learning-Whittle. Simulation results show that the proposed policies yield excellent empirical performance.

\end{abstract}

%% file: intro.tex
\section{Introduction}\label{sec:intro}

An increasing number of devices such as smart wearables, mobile phones, are trending towards various high data-rate services such as video streaming, web browsing, software downloads. This has urged service providers, e.g., content delivery network (CDN), to pursue new service technologies that yield a good quality of experience. One such technology entails network densification by deploying \textit{edge servers} and each of which is empowered with a small base station (SBS), e.g. the storage-assisted future mobile Internet architecture and cache-assisted 5G systems~\cite{andrews2012femtocells}.  Service requests are often derived from dispersed sources, for example customers, Internet-of-Thing (IoT) devices, embedded systems, etc. As a result, the distributed infrastructure of such a storage-assisted system operator (e.g., Akamai CDN, Google Analytics, etc) typically has a hub-and-spoke model as illustrated in Figure~\ref{fig:system}.

In such systems, customers send service requests to geographically distributed edge servers near them. If the requested service is available with the edge, then it is satisfied with a negligible latency. Otherwise, the request is sent to the central service warehouse (CSW, see Figure~\ref{fig:system}), which then provides the service but the resulting latency is higher. 
While the CSW is often located in a well-provisioned data center, resources are typically limited at edges, i.e., the capacity of edge servers is usually limited. 
Furthermore, the network interface and hardware configuration of edge servers might greatly impact the response of services \cite{chen2017empirical}. These issues pose significant challenges to ensuring that customers waiting at the network edge for a service face minimal delay. A major challenge in such storage-assisted network edges is to ensure that the delays faced by customers that wait at the network edges is ``minimal''. In this work, we propose to use dynamic policies~\cite{ross2013applied,ross2014introduction} in order to make decisions regarding which set of services should be placed at each edge of the network so as to minimize the cumulative delays of customers. We refer to the problem of making optimal dynamic decisions regarding service placement as the \textit{optimal service placement problem}. 

\begin{figure}
	\centering
	\includegraphics[width=0.35\textwidth]{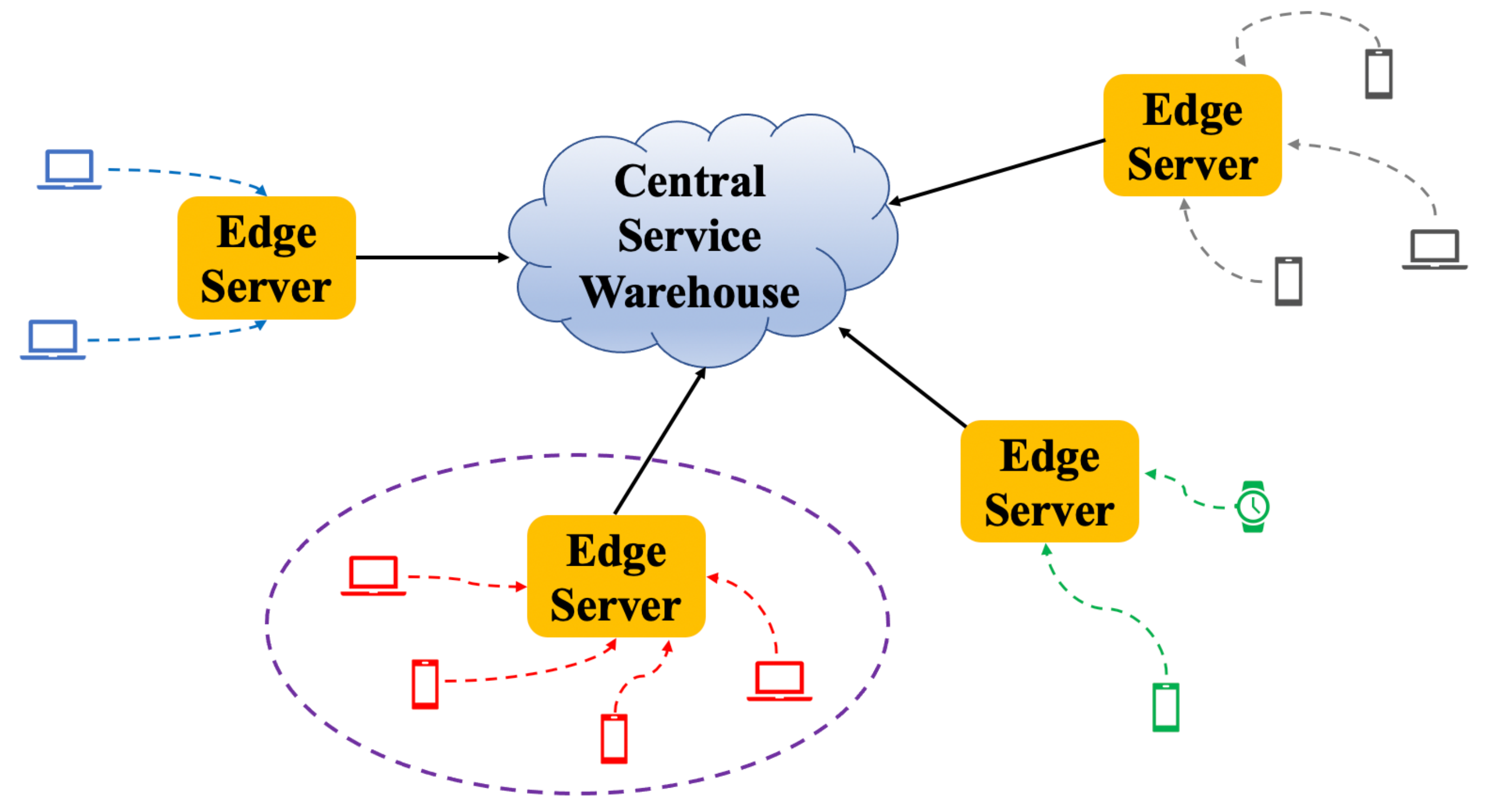}
	\vspace{-0.05in}
	\caption{The distributed model for a typical service placement comprises a single central service warehouse and multiple edge servers connected by cheap backhaul links. }
	\label{fig:system}
	\vspace{-0.25in}
\end{figure}
Existing works e.g., \cite{he2018s,farhadi2019service,poularakis2019joint,pasteris2019service,lin2020service} that address the issue of placing services at the network edges so as to minimize the customer delays suffer from severe limitations. For example, they mostly restrict their analysis to a deterministic system model, so that the resulting service placements are highly pessimistic. Secondly, their solutions often assume that the parameters which describe the system dynamics, for example service request rates, service delivery rates etc., are fixed and known to the service provider. In reality, these quantities are typically unknown and possibly also time-varying.

We instead model the operation of network edges by utilizing an appropriate stochastic framework, and derive easily implementable policies that make decisions regarding which service should be placed at edges. Secondly, when the system parameters are unknown, we derive novel machine learning (ML) techniques in order to efficiently learn the unknown network system parameters and make optimal service placement decisions dynamically. This is particularly important since with the advent of cost-effective ML solutions, the network operator can deploy them in real-time in order to optimize the system performance. In particular, we raise the following question: \emph{Can we leverage ML for maximizing the benefits of storage resources available at the edges and optimizing the performance of service placement at the network edge? }

\subsection*{Main Results}
We consider the setup where there is a single network edge with a capacity of $K$ units shared by $N$ services. New requests for each service, i.e. customers that demand this service, are modeled by Poisson processes, while the time taken to deliver services to customers is taken to be exponential random variable.  A decision maker has to choose which $K$ out of these $N$ services to be placed at the edge at each time $t$.  

We first focus on the scenario when the system parameters, e.g. arrival rates and service delivery rates are known to the decision maker.  We pose the problem of placing services at the network edge so as to minimize the cumulative delays faced by customers waiting in queue as a Markov decision process (MDP)~\cite{puterman2014markov,ross2013applied,ross2014introduction} in which the system state is given by describing, for each service, the number of customers that are currently waiting at the edge to obtain this service.  This MDP is intractable in general due to the curse of dimensionality~\cite{bellmanbook,bertsekas1995dynamic}, i.e., the size of state-space and computational complexity of obtaining optimal solution grows exponentially with $N$, and also since the optimal solution is complicated, has a complex structure that depends on the entire state-space description, and hence not easily implementable. 

Our MDP can be viewed as a special instance of the Restless Multi-Armed Bandit (RMAB) \cite{whittle1988restless,gittins2011multi}, in which we can view each service as a bandit, and its queue length (the number of customers waiting to receive this service) as the state of the bandit. Using Little's law~\cite{john1961little,kleinrock1976queueing}, delay minimization is seen to be equivalent to minimizing the cumulative average queue lengths. Thus, the problem becomes equivalent to minimizing the costs (maximizing the negative value of queue lengths) earned by ``pulling arms of bandits'' (i.e., placing services at edge) under the constraint that the total number of services placed at the edge is less than the edge capacity. However, the RMAB based formulation in general suffers the curse of dimensionality and is provably hard \cite{papadimitriou1994complexity}. 

The Whittle index policy~\cite{whittle1988restless} is computationally tractable, and known to be asymptotically optimal~\cite{verloop2016asymptotically} for RMAB. We propose to use it in order to make optimal service placement decisions.  We show that our MDP is indexable, and derive explicitly the Whittle indices for each service as a function of the number of customers that are waiting to receive it.

Since the system parameters, e.g. service request arrival and delivery rates etc., are typically unknown and time varying, we further explore the possibility of designing efficient learning augmented algorithms to address these challenges. Though we could use reinforcement learning (RL) algorithms to resolve this issue, the learning regret\footnote{Loosely speaking, the learning regret of a learning algorithm is the rate with which its decisions converge to the optimal ones.}~\cite{auer2002finite,bubeck2012regret,lattimore2020bandit} of the resulting solution
scales linearly with the size of the state space~\cite{jaksch2010near}, and hence it would be too slow to be of any practical use. We thus derive efficient RL algorithms that utilize the structure of the service placement problem in order to yield a low learning regret. 

The first such algorithm entitled Upper Confidence Bound (UCB)-Whittle relies upon the principle of optimism in the face of uncertainty \cite{lai1985asymptotically,auer2002finite,mete2020reward}.  UCB-Whittle combines the asymptotic optimality property of the Whittle policy~\cite{verloop2016asymptotically} with the ``efficient exploratory behavior'' of the UCB-based~\cite{auer2002finite} learning algorithms. Thus, UCB-Whittle maintains a confidence ball that contains the true (unknown) parameter with a high probability.  Within an RL episode it uses the Whittle index policy based on an optimistic estimate of the parameter from within this confidence ball. Since the computational complexity of deriving Whittle index policy scales linearly with the number of services, it can be easily implemented, and moreover we also show that its learning regret is low. 

Our second algorithm entitled Q-learning-Whittle, utilizes Q-learning iterations for each service by using a two time-scale stochastic approximation.  It leverages the threshold-structure of the optimal policy to learn the state-action pairs with a deterministic action at each state.  Hence, Q-learning-Whittle only learns Q-value of those state-action pairs following the current threshold policy.  This new update rule leads to a substantial enhancement over the sample efficiency of Q-learning.  Our numerical evaluations show that the proposed policies yield excellent empirical performance.

The rest of the paper is organized as follows.  We describe the system model and the MDP formulation in Section~\ref{sec:formulation}. We present the Whittle index policy in Section~\ref{sec:lagrangian-relaxation}.  We design novel learning augmented index policies UCB-Whittle and Q-learning-Whittle in Section~\ref{sec:learning}. The proofs of their performance are given in Section~\ref{sec:learning-proof}. Numerical results are presented in Section~\ref{sec:sim}. We discuss related work in Section~\ref{sec:related} and conclude the paper in Section~\ref{sec:conclusion}.  Additional proof details are presented in Appendix~\ref{sec:app}.

%% file: model.tex
\section{System Model}
\label{sec:formulation}

Consider a heterogeneous network as shown in Figure~\ref{fig:system} in which there are multiple edge servers that are geographically distributed and provide services to customers. Each single edge server is assigned an area, and customers that fall within this area send their requests for various services to it.  When a customer makes request for a service to its edge server, and if this service is available with the edge, then it is delivered to the customers directly.  Otherwise, the request is sent to the CSW, which then provides service to the customers, though at the cost of a longer latency.  Since a single service can be replicated effortlessly at multiple edges, we focus exclusively on the scenario when there is only a single edge server.  There is a single service provider that provides $N$ distinct services that are indexed by $i\in\mathcal{N}=\{1,\cdots, N\}$.  These are provided to customers via the edge server.  For simplicity, we assume that all services are of unit size\footnote{Our model can be easily generalized to the case when different services are of varied sizes.}. The capacity of the edge server is equal to $K$ units, where $K\in\mathbb{N}_+$, and $K<N$. Thus, at any given time, only $K$ out of $N$ services can be placed at the edge server.

\begin{figure}
	\centering
	\includegraphics[width=0.35\textwidth]{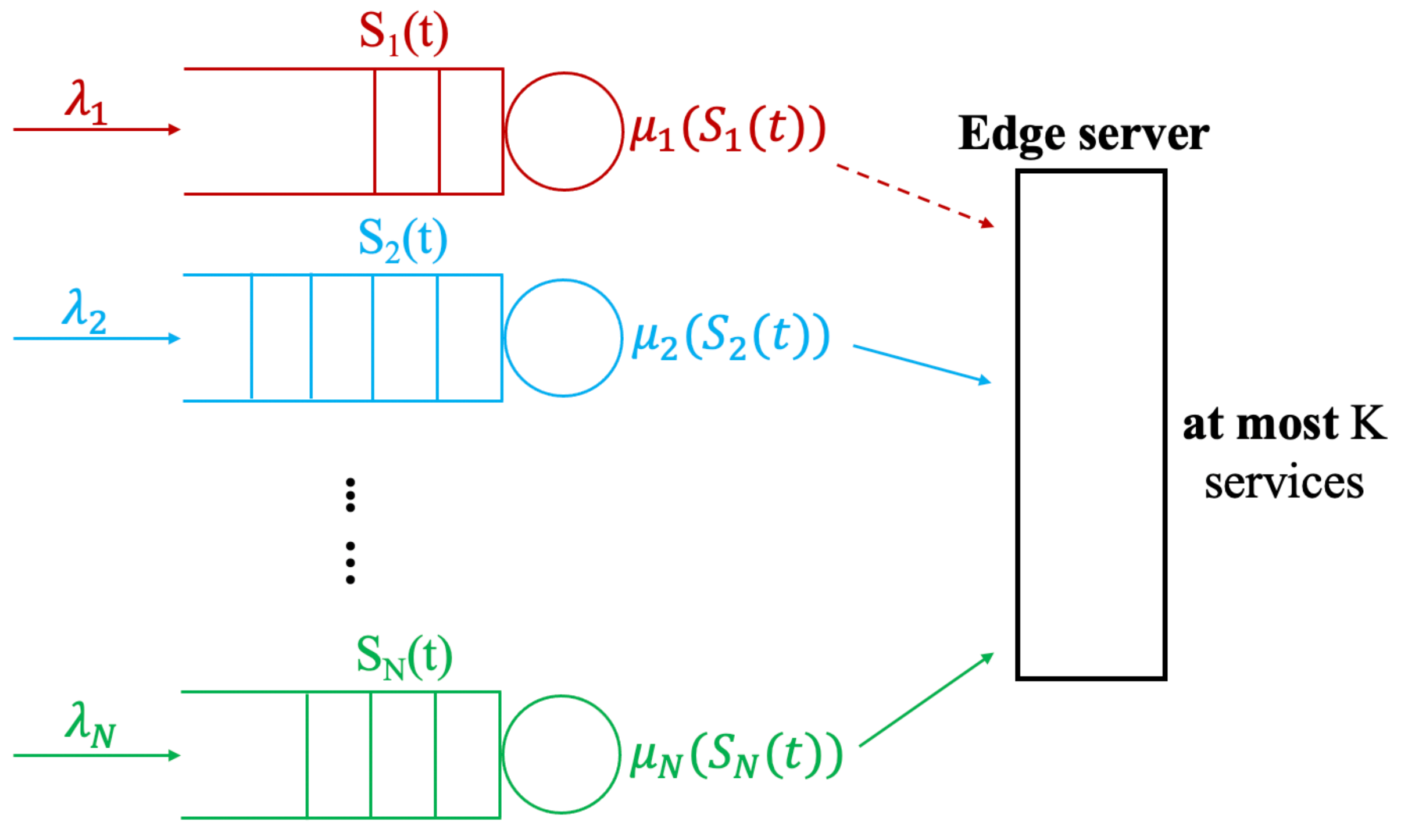}
	\vspace{-0.05in}
	\caption{A service placement system where the edge server can host at most $K$ services at any moment in time.}
	\label{fig:queue}
	\vspace{-0.2in}
\end{figure}

We now formulate the problem of making decisions regarding which service should be placed at the edge, as an MDP~\cite{puterman2014markov}. We discuss the system states, actions, rewards and controlled transition probabilities for this MDP.

Requests for service $i\in\mathcal{N}$ from customers arrive at the edge server according to a Poisson process with arrival rate $\lambda_i$. The time taken to deliver service $i$ to a customer is a random variable that is exponentially distributed with mean $1/\mu_i$. The delivery times are independent across customers and services. The number of outstanding requests (customers) for $i$ at time $t$ are denoted by $S_i(t)$. Denote $\vec{S}(t):=(S_1(t),\cdots,S_N(t))$. Let $A_i(t)\in \{0,1\}$ be the decision for service $i$ at time $t$, i.e., $A_i(t)=1$ means that service $i$ is placed at edge at time $t$, while $A_i(t)=0$ means that it is not available at edge.  Denote $\vec{A}(t):=(A_1(t),A_2(t),\ldots,A_N(t))$. In accordance with the terminology for RMAB, when $A_i(t)$ is $1$, we say that bandit $i$ is active, while when $A_i(t)$ is $0$, we say that it is passive. Decisions are made only at those time instants when either a new request arrives, or a service delivery occurs, so that $\vec{A}(t)$ stays constant in between these instants.

The state of the $i$-th queue can change from $S_i$ to either $S_{i}+1$ or $(S_i-1)_+$.  The transition rates are as follows: 
\begin{align}\label{eq:transition}
\begin{cases}
\vec{S}\rightarrow\vec{S}+\vec{e}_i,\quad\text{with transition rate $b_i(S_i, a_i)$},\\
\vec{S}\rightarrow\vec{S}-\vec{e}_i,\quad\text{with transition rate $d_i(S_i, a_i)$},
\end{cases}
\end{align}
where $\vec{e}_i$ is a $N$-dimensional vector with all zero elements except the $i$-th entry being equal to $1$, while $b_i(\cdot,\cdot),d_i(\cdot,\cdot)$ are the birth and death rates. These rates are as follows:
\begin{align}
b_i(S_i, a_i)&=\lambda_i,\quad
d_i(S_i, a_i)=\mu_i(S_i) a_i,
\end{align}
where the function $\mu_i$ satisfies $\mu_i(0)=0$. As is clear from the definitions above, we allow the death rates to be a function of the state. This allows us to model much more realistic scenarios. In this paper, we assume the departure rate  as the classic $M/M/k$ queue for simplicity, i.e., {$d_i(S_i, a_i)=\mu_i S_ia_i.$} This means that an adequate amount of bandwidth is available for connecting the edge to customers, so that in comparison with the capacity constraint of $K$ units at the edge server, the last-hop transmission (edge to customers) is not the bottleneck.

We define
\begin{align}\label{eq:cost}
C_i(S_i,a_i) := S_i/\lambda_i,
\end{align}
to be the instantaneous cost incurred when state of service $i$ is $S_i$, and action $a_i$ is applied to it.  As is shown below, by Little's Law, the average latency is equivalent to the average total number of outstanding requests over the request arrival rate in the system.  Hence the cost represents the average latency for obtaining the service from the edge server for customers.

Let $\mathcal{F}_t$ denote the sigma-algebra~\cite{shiryaev2007optimal} generated by the random variables $\left\{(\vec{S}(\ell),\vec{A}(\ell) ):0\le \ell < t\right\}$. A policy $\pi$, is a collection of maps $\mathcal{F}_t\mapsto \vec{A}(t), t=1,2,\ldots$ that yields a feasible allocation $\vec{A}(t),\sum_i A_i(t)\le K$ based on the past history of system operation. The decisions $\vec{A}(t)$ determine which services should be placed on the edge, i.e., which bandits should be made active.

The expected average service delivery latency, or the average cost by Little's Law 
to customers under policy $\pi$ is then given as follows,
\begin{align}\label{eq:average-latency}
\bar{C}_\pi :=\limsup_{T\rightarrow\infty}\sum_{i=1}^N\frac{1}{T}\mathbb{E}_{\pi}\left(\int_{0}^T C_i(S_i(t))dt\right),
\end{align}
where $\mathbb{E}_{\pi}$ denotes that the expectation is taken with respect to the measure induced by policy $\pi$. Our goal is then to derive a policy $\pi$ to minimize the average delivery latency under an edge server capacity constraint, which can be formulated as the following MDP:
\begin{align}\label{eq:obj-continous} 
\min_{\pi\in\Pi} \quad&\bar{C}_{\pi}, \quad 
\text{s.t.} \quad\sum_{i=1}^N A_i(t)\leq K, \quad\forall t.
\end{align}
The following result states the existence of a value function and average cost value for the MDP~\eqref{eq:obj-continous} (Theorem 6.3.1, Section 6.3.2 \cite{puterman2014markov}).
\begin{lemma}
It is known that there exists $f$ and $V(\cdot)$ that satisfy the Dynamic Programming equation
\begin{align}\label{eq:obj-continous-DP}
f&=\min_{{\sum_i a_i\leq K}}\Bigg( \sum_{i=1}^N\Bigg[ C_i(S_i, a_i)+b_i(S_i, a_i) V(\vec{S}+e_i)+ \nonumber\displaybreak[0]\\
&d_i(S_i, a_i) V(\vec{S}-e_i)\!-\!(b_i(S_i, a_i)\!+\!d_i(S_i, a_i) ) V(\vec{S})   \Bigg]\!\Bigg),
\end{align}
where a stationary policy that realizes the minimum in~(\ref{eq:obj-continous-DP}) is optimal with $f=\min_{\pi}C_{\pi}$ and $V(\vec{S})$ being the value function  \cite{puterman2014markov}. An optimal policy ~(\ref{eq:obj-continous-DP}) can be obtained numerically by using value iteration or policy improvement algorithms.
\end{lemma}
Though in principle one could use the iterative algorithms mentioned above, in reality the ``curse of dimensionality", i.e. exponential growth in the size of the state-space with the number of services $N$, renders such a solution impractical. Thus, we resort to using the Whittle index policy that is computationally appealing.

%% file: relaxation.tex
\section{Service Placement using Whittle Index Policy} \label{sec:lagrangian-relaxation}

In general, even if the state-space were bounded (say for example by truncating the queue lengths), the MDP~(\ref{eq:obj-continous}) is a hard problem to solve because the optimal decisions for $N$ services are strongly coupled. We realize that problem~(\ref{eq:obj-continous}) can be posed as a RMAB problem in which the queue length $S_i(t)$ is the state of bandit $i$. A tractable solution to RMAB, with a computational complexity that scales linearly with $N$, is the Whittle index policy~\cite{whittle1988restless}. We briefly describe the notion of indexability and the Whittle index policy and then show that our problem is indexable.  For the ease of exposition, the proofs of main results in this section are relegated to Appendix~\ref{sec:app-non-learning}.

\subsection{Whittle Index Policy}
Consider the following problem, which arises from~\eqref{eq:obj-continous} by relaxing its constraint to time-average:
\begin{align}\label{eq:capacity-continous-average}
\min_{\pi\in\Pi}&~\bar{C}_\pi,  \quad
\text{ s.t.}\limsup_{T\rightarrow\infty} \frac{1}{T}\mathbb{E}_{\pi}\left(\int_{0}^T\sum_{i=1}^N A_i(t)dt\right)\leq K.
\end{align}
It can be shown that the relaxed problem~(\ref{eq:capacity-continous-average}) is computationally tractable, has a computational complexity that scales linearly with $N$, and has a decentralized solution in which the decisions for each service are made only on the basis of the state of that service. Consider the Lagrangian associated with this problem,
\begin{align}\label{eq:obj-continous-relaxed}
\cL(\pi,W) :=
\limsup_{T\rightarrow\infty} &\frac{1}{T}\mathbb{E}_{\pi}\Bigg\{\int_0^T\bigg(  \sum_{i=1}^N C_i(S_i(t))  \nonumber\displaybreak[0]\\
& \qquad {-W}\bigg[ K- \sum_{i=1}^N A_i(t) \bigg] \bigg)  dt      \Bigg\},
\end{align}
where $W$ is the Lagrangian multiplier, and $\pi$ is the allocation policy. Also define the associated dual problem,
\begin{align}
D(W) := \min_{\pi\in \Pi}  \cL(\pi,W).
\end{align}
Since the Lagrangian decouples into the sum of $N$ individual service MDPs, it turns out that in order to evaluate the dual function at $W$, it suffices to solve $N$ single-service MDPs of the following form ($\pi_i$ is a policy for service $i$):
\begin{align}\label{def:single_mdp}
\min_{\pi_i} \quad& \bar{C}_{\pi_i, i},
\end{align}
where
\begin{align}\label{def:single_mdp1}
\bar{C}_{\pi_i, i}: &=\limsup_{T\rightarrow\infty} \frac{1}{T}\mathbb{E}_{\pi_i}\Bigg\{\int_0^T\bigg( C_i(S_i(t))   \nonumber\\
&\qquad\qquad\qquad\qquad\qquad {-W}(1-A_i(t))           \bigg)  dt      \Bigg\}.
\end{align}

\begin{definition}(Indexability):
Consider the single-service MDP~\eqref{def:single_mdp} for service $i$. Let $D_i(W)$ denote the set of those states $s$ for which the optimal action is to choose action $a=0$ (passive). Then, the $i$-th MDP is indexable if the set $D_i(W)$ increases with $W$, i.e., if $W> W^\prime$ then $D_i(W)\supseteq  D_i(W^{\prime})$. The original MDP~\eqref{eq:obj-continous} is indexable if each of the $N$ single-service MDPs are indexable.
\end{definition}

\begin{definition}(Whittle Index)
If the single-service MDP for service $i$ is indexable, then the Whittle index in state $s$ is denoted $W_i(s)$, and is given as follows:
\begin{align*}
W_i(s):=\inf_{W\ge 0}\{ s \in D_i(W) \}.
\end{align*}
Thus, it is the smallest value of the parameter $W$ such that the optimal policy for service $i$ is indifferent towards $a=0$ and $a=1$ when the state is equal to $s$.
\end{definition}

\begin{definition}(Whittle index rule)
At each time $t$, the Whittle index rule prioritizes the bandits (services) in the decreasing order of their Whittle indices $W_i(S_i(t))$. It then activates $K$ bandits (places $K$ services) that have the highest priority.
\end{definition}
The Whittle index policy is in general not an optimal solution to the original problem~\eqref{eq:obj-continous}.  However, it has been proved that Whittle index policy is asymptotically optimal \cite{weber1990index,verloop2016asymptotically,ouyang2012asymptotically} as the number of bandits is scaled up while keeping constant their relative population sizes and the
number of bandits that can be activated simultaneously.

\subsection{MDP~\eqref{eq:obj-continous} is indexable}
Our proof of indexability relies on the ``threshold'' property of the optimal policy for the single-service MDP, i.e., the service is placed on the edge server only when the number of requests for it at the edge is above a certain threshold.
\begin{proposition}\label{prop:threshold-policy}
Fix a $W\ge 0$, and consider the single-service MDP~\eqref{def:single_mdp}.  The optimal policy for this problem is of threshold type. The threshold depends upon $W$.
\end{proposition}

We now compute the stationary distribution of a threshold policy as a function of its threshold. This result is useful while proving indexability of the original MDP.
\begin{proposition}\label{prop:stationary-distributions}
The stationary distribution of the threshold policy $\pi=R$ satisfies
\begin{align}
&q_i^R(R^\prime)=0,\nonumber\\
&q_i^R(R)=1/\left(1+\sum\limits_{j=1}^{\infty}\left(\frac{\lambda_i}{\mu_i}\right)^j\frac{1}{\Pi_{k=1}^j(R+k)}\right)\\ \nonumber
&q_i^R(R+l)=\left(\frac{\lambda_i}{\mu_i}\right)^l\frac{1}{\Pi_{k=1}^l(R+k)}q_i^R(R), l=1,2,\ldots,
\end{align}
 where $R^\prime$ is a dummy states representing state $0$ to $R$.  For the notation abuse, we use a superscript $R$ to denote the stationary distribution under a particular policy $R.$
\end{proposition}

We conclude this subsection by show that all bandits are indexable under our model.

\begin{proposition}\label{prop:indexable}
The MDP~\eqref{eq:obj-continous} is indexable.
\end{proposition}

We are now ready to derive the Whittle indices for the MDP~\eqref{eq:obj-continous}.
\begin{proposition}\label{prop:whittle-index-closed}
The Whittle index is given by
\begin{align}\label{eq:whittle-index1}
W_i(R)=\frac{\mathbb{E}_{R+1}[C_i(S_i)] - \mathbb{E}_{R}[C_i(S_i)] }{\sum_{S_i=0}^{R+1} q_{i}^{R+1}(S_i)-\sum_{S_i=0}^{R} q_{i}^R(S_i) },
\end{align}
when the right hand side of~(\ref{eq:whittle-index1}) is non-decreasing in $R$.
\end{proposition}

\begin{remark}
If~(\ref{eq:whittle-index1}) is non-monotone in $R$, then the Whittle index cannot be derived by equating the average cost of two consecutive threshold policies.  Instead, the algorithm described in \cite{glazebrook2009index,larranaga2015asymptotically} provides a means to obtain Whittle index.  In particular, the threshold policy $-1$ is compared to an appropriate threshold policy $R_0>-1$, we obtain $W(R_0).$ Whittle index for all thresholds $r\in\{0, 1,\cdots, R_0\}$ equals $W(R_0).$  The right choice of $R_0$ is the result of an optimization problem.  Similarly, the index in state $R_0+1$ is computed by comparing threshold policy $R_0$ to an appropriate threshold policy $R_1>R_0$. Then the Whittle index is $W(R_1)$ for all $r\in\{R_0+1, \cdots, R_1\}$.  The algorithm terminates when $R_i=\infty$ in one iteration $i.$
 \end{remark}

 \begin{remark}
 Since the cost function and the stationary probabilities are known, ~(\ref{eq:whittle-index1}) can be numerically computed.
From~(\ref{eq:whittle-index1}), it is clear that the index of bandit $i$ does not depend on the number of requests to other services $j$, $j\neq i.$  Therefore, it provides a systematic way to derive simple policies that can be easily implementable.
\end{remark}

%% file: learning.tex
\section{Reinforcement Learning Based Whittle Index Policy}\label{sec:learning}

The computation of Whittle indices requires us to know the arrival rates $\lambda_i$ and service rates $\mu_i$.  Since these quantities are typically unknown and possibly time-varying, the assumption that these are known to the service provider, is not practical.  Hence, we now develop learning augmented algorithms that can ensure almost the same performance as the Whittle index policy.  In particular, we propose two RL algorithms that can fully exploit the structure of our proposed Whittle index policy.  
For the ease of exposition, we present our algorithms and main results in this section, and relegate the proofs to Section~\ref{sec:learning-proof}.

\subsection{UCB-Whittle}\label{sec:learning-ucb}
We first propose a UCB type algorithm based on the principle of optimism in the face of uncertainty~\cite{auer2002finite,lattimore2020bandit,mete2020reward}.
To simplify the exposition, we will sample our continuous-time system at those discrete time-instants at which the system state changes, i.e. either an arrival or a departure occurs. This gives rise to a discrete-time controlled Markov process~\cite{puterman2014markov}, and henceforth we will work exclusively with this discrete time system. We use $m\lmu:=1\slash \mu_{i}$ to denote the mean delivery time of service $i$, and $m\llm:=1\slash \lambda_{i}$ to denote the mean inter-arrival times for requests for service $i$. In what follows, we will parameterize the system in terms of the mean inter-arrival times and mean delivery times of all services. This is useful since the empirical estimates of mean inter-arrival times and mean delivery times are unbiased, while the empirical estimates of mean arrival rates and delivery rates are biased. Hence, this parametrization greatly simplifies the exposition and analysis. We denote $\te_{i,0} = (m\llm,m\lmu)$, and $\te_0=\left\{ \te_{i,0}\right\}_{i=1}^N$. Thus $\te_0$ denotes the vector comprising of the true system parameters.

Let $\vec{s},\vec{s}~^{\prime}$ be two possible system state values, and {$\vec{a}$} be an allocation for the combined system composed of $N$ services. $p_{\te}(\vec{s}~\up |\vec{s}, \vec{a})$ denotes one-step controlled transition probability for the combined system from state $\vec{s}$ to state $\vec{s}~\up$ under allocation $\vec{a}$ when the system parameter is $\te$.  Since Whittle indices $W_i(\cdot)$ are also a function of the parameter $\te_i$, we denote $W_{i,\te_i}(\cdot)$ in order to depict this dependence. For $\te = \{\te_i\}_{i=1}^{N}$, we use $W_{\te}:=  \left\{W_{i,\te_i}(\cdot)\right\}_{i=1}^N$, and by notational abuse we also use $W_{\te}$ to denote Whittle index policy that uses Whittle indices $\left\{W_{i,\te_i}(\cdot)\right\}_{i=1}^N$ for placing services.

The learning algorithm knows that $\te_0$ belongs to a finite set $\Theta$. For a stationary control policy $\pi$ that chooses $\vec{A}(t)$ on the basis of $\vec{S}(t)$, we let $\bar{C}(\pi;\te)$ denote its infinite horizon average expected cost (i.e., latency) when true system parameter is equal to $\te$, and $\bar{C}(\pi;\te,H)$ denote the cumulative cost incurred during $H$ time-steps. As is common in the RL literature~\cite{jaksch2010near,lattimore2020bandit,mete2020reward}, we define the ``gap'' $\Delta$ as follows,
\begin{align}\label{def:gap}
\Delta := \bar{C}(W_{\te_0};\te_0) - \max_{\te\in \Theta: \bar{C}(W_{\te};\te_0) \neq \bar{C}(W_{\te_0};\te_0) } \bar{C}(W_{\te};\te).
\end{align}
Throughout, we assume that we are equipped with a probability space $(\Omega,\mathcal{F},\bP)$~\cite{resnick2019probability}.

\subsubsection{Learning Setup}
We assume that the algorithm operates for $T$ time steps. We consider an ``episodic RL'' setup~\cite{sutton2018reinforcement} in which the total operating time horizon of $T$ steps is composed of multiple ``episodes''. Each episode is composed of $H$ consecutive time steps. Let $\tau_k:= (k-1)H, k=1,\ldots$ denote the starting time of the $k$-th episode, and $\cE_k$ the set of time-slots that comprise the $k$-th episode. We assume that {the system state is reset to $\vec{0}$ at the beginning of each episode}. This can be attained by discarding, at the end of each episode, those customers from our system which have not received service by the end of the current episode. Denote by $\mathcal{F}_t$ the sigma-algebra~\cite{shiryaev2007optimal} generated by the random variables $\left\{(\vec{S}(\ell),\vec{A}(\ell)):\ell=1,2,\ldots,t\right\}$. A learning rule $\phi$ is a collection of maps $\mathcal{F}_t\mapsto \mathcal{A}$ that utilizes the past observation history in order to make service placement decisions for the current time-step $t$, where $\mathcal{A}$ denotes the action space.

\subsubsection{Learning Rule}
Let $\hat{m}\llm(t),\hat{m}\lmu(t)$ denote the empirical estimates of the mean arrival times and {mean delivery time} of service $i$ at time $t$, i.e., 
\begin{align}
    \hat{m}\llm(t) \!\!:&\!\!=\!\! \frac{\text{Total time taken for arrivals for $i$ until $t$}}{ \text{Total requests until $t$}}, \label{est:1}\displaybreak[0]\\
     \hat{m}\lmu(t) \!\!:&\!\!=\!\!  \frac{\text{Total time spent delivering $i$ until $t$ }}{\text{ Number of deliveries for $i$ until $t$  }}.\label{est:2}
\end{align}
Denote $\hat{\te}_i(t):= (\hat{m}\llm(t),\hat{m}\lmu(t))$, and construct the confidence interval associated with $\hat{\te}_i(t)$ as follows,
\begin{align}\label{def:ci_i}
    \cO_i(t) &:=\{(m_1,m_2): |m_1 - \hat{m}\llm(t)|\le\epsilon\llm(t), \nonumber\displaybreak[0]\\
    &\qquad\qquad\qquad\qquad  | m_2 - \hat{m}\lmu(t) |\le \epsilon_{i,\mu}(t)\},
\end{align}
where $\epsilon\llm(t),\epsilon_{i,\mu}(t)$ denote the radii of the confidence intervals, and are given as follows,
\begin{align}\label{def:radius}
    \epsilon\llm(t):&= \sqrt{\frac{K_1}{N\llm(t)}\log\left(\frac{N\cdot T^{b}}{\delta} \right)},\displaybreak[0]\\
    \epsilon\lmu(t) :&= \sqrt{\frac{K_1}{N\lmu(t)}\log\left(\frac{N\cdot T^{b}}{\delta} \right)},
\end{align}
where $\delta>0$ is a user-specified parameter, $N\llm(t) (N\lmu(t))$ denotes the number of samples obtained for estimating the arrival rate (delivery rate) of service $i$ until time $t$, and $b,K_1>0$ are constants.  We will use $\delta = \frac{1}{T}$ in order to ensure that the learning algorithm has a good performance (see Theorem~\ref{th:ucb_regret}).

The confidence ball of the overall system is denoted as\footnote{Though we define confidence intervals and empirical estimates for each time, they are used only at  times $\tau_k$ that mark the beginning of a new episode. }
\begin{align}\label{def:ci}
    \cO(t):= \left\{ \te=\{\te_i\}_{i=1}^{N} \in \Theta: \te_{i} \in \cO_i(t), \forall i \right\}.
\end{align}
The learning rule then derives an ``optimistic estimate'' $\tilde{\te}_k$ of the true parameter $\te_0$ as follows,
\begin{align}\label{def:tilde_te}
    \tilde{\te}_k \in \arg\min_{\te\in \cO(\tau_k)} \bar{C}(W_{\te};\te).
\end{align}
In case $\cO(\tau_k)$ is empty, then $\tilde{\te}_k$ is chosen uniformly at random from the set $\Theta$. During $\cE_k$, the learning rule implements $W_{\tilde{\te}_{k}}$.  We summarize this as Algorithm~\ref{alg:ucb-whittle}.

\begin{algorithm}
\caption{UCB-Whittle}
\begin{algorithmic}
\State Initialize: $\hat{\lambda}_i(0)= .5,~\hat{\mu}_i(0) = .5, N_i(0)=0$ for all $i=1,2,\ldots,N$.
\For{Episodes $k=1,2,\ldots,$}
\State $k\to k+1$.
\State Update the estimate $\hat{\te}(\tau_k)$ by using ~\eqref{est:1}-\eqref{est:2} $\forall i$.
\State Calculate $\cO_i(t), \forall i$ according to~\eqref{def:ci_i}.
\State Calculate optimistic estimate $\tilde{\te}_k$ by solving~\eqref{def:tilde_te}.
\State Calculate Whittle index policy $W_{\tilde{\te}_k}$ and implement $W_{\tilde{\te}_k}$ in the current episode $\cE_k$.
\EndFor
\end{algorithmic}
\label{alg:ucb-whittle}
\end{algorithm}

\subsubsection{Learning Regret}
We define the regret $R(\phi,T)$ of the learning rule $\phi$ as follows,
\begin{align}\label{def:regret}
    R(\phi,T) :=\sum_{k=1}^{K}
    \left[ \sum_{t\in\cE_k} c(\vec{s}(t),\vec{u}(t))  -  \bar{C}(W_{\te_0};\te_0,H)\right],
\end{align}
where $\bar{C}(W_{\te_0};\te_0,H)$ is the expected value of the cost incurred during $H$ time-slots when $W_{\te_0}$ is used on the system that has parameter equal to $\te_0$. Our goal is to design a learning rule that minimizes the expected value of the regret $R(\phi,T)$. Thus, our benchmark is the Whittle index rule that is designed for the system with parameters $\te_0$. We next present the regret of UCB-Whittle algorithm.

\begin{theorem}\label{th:ucb_regret}
Consider the UCB-Whittle Algorithm applied to adaptively place services on the edge server so as to minimize the average latency faced by customers. We then have that with a probability greater than $1-\left(\delta + \epsilon + \frac{\delta}{T}\right)$, the regret $R(T)$ is bounded by $\frac{4}{\eta H \epsilon^2(\Delta)}\left[\frac{-\log \epsilon}{\ell b}\right]^{2}\log\left(\frac{N T^{b}}{\delta} \right)$, where $\epsilon>0$ is a user-specified parameter, $\ell b$ is a lower bound on the mean inter-arrival times and delivery times, $\Delta$ is the gap~\eqref{def:gap}, and the quantity $\epsilon(\Delta)$ is as in Lemma~\ref{lemma:exp_concentration}.  With $\epsilon = \frac{1}{T}, \delta = \frac{1}{T}$, we obtain that the expected regret $\bE\left( R(T)\right)$ of UCB-Whittle is upper-bounded as follows, \begin{align*}
\bE\left( R(T)\right) \le \frac{4}{\eta H \epsilon^2(\Delta)}  \frac{\log^{2}(T)\log\left(N T^{b+1}\right)}{\min_i \min \left\{ 1\slash \lambda_i, 1\slash \mu_i \right\}  }.
\end{align*}
\end{theorem}
\textit{Proof Sketch:} We show that the sample path regret can be upper-bounded by the number of episodes in which $\tilde{\te}_k$ is not equal to $\te$, i.e., sub-optimal episodes. We then decompose the sample space $\Omega$ into ``good set'' $\cG$ and ``bad set'' $\cG^c$. Loosely speaking, on $\cG$ we have that each customer is sampled sufficiently many times, and moreover the sample estimates $\hat{\te}(t)$ concentrate around the true value $\te_0$. We show that the good properties of $\cG$ imply a bound on the number of sub-optimal episodes, and hence the regret. The expected regret on $\cG^c$ is bounded primarily by bounding its probability.   The proof details are presented in Section~\ref{sec:learning-proof}.

\begin{remark}
It is well known by now that the learning regret of commonly encountered non-trivial learning tasks grows asymptotically as $O(\log T)$~\cite{lai1985asymptotically,lattimore2020bandit,auer2002finite}. A key research problem is to design a learning algorithm that yields the ``lowest possible'' pre-factor to this logarithmic regret. Theorem~\ref{th:ucb_regret} shows that UCB-Whittle has a regret $O(\log^2 T)$. An interesting question that arises is whether we could attain $O(\log T)$ regret with computationally tractable algorithms, and also what is the ``optimal'' instance-dependent~\cite{lattimore2020bandit} pre-factor.
\end{remark}

\subsubsection{Computational Complexity}
Problem~\eqref{def:tilde_te} needs to be solved at the beginning of $k$-th episode. Solving~\eqref{def:tilde_te} requires us to evaluate the average cost of the Whittle policy $W_{\te}$ for each possible system parameter $\te\in\Theta$. However, deriving the average cost of Whittle policy is computationally expensive because of the curse of dimensionality. However, we note that the Whittle index policy is asymptotically optimal as the population size is scaled up to $\infty$.  Moreover, it is shown in~\cite{verloop2016asymptotically} that the limiting value of the normalized\footnote{Normalization involves dividing the expected cost by the population size.} average cost is equal to the value of the allocation problem in which the hard constraints are relaxed to time-average constraints. This means that for large population size we could approximate the cost $\bar{C}(W_{\te};\te)$ by the optimal value of the relaxed problem. The relaxed problem is tractable since the decisions of different services are decoupled, and hence the computational complexity scales linearly with the number of {services}.  Moreover, since a threshold policy is optimal for each individual service, and the stationary probability for such a single service that employs a threshold policy, is easily computed, the quantity $\bar{C}(W_{\te};\te)$ can be obtained easily.

\input{Qlearning}

%% file: Qlearning.tex
\subsection{Q-learning-Whittle}
 
 We then design a heuristic algorithm for learning Whittle indices based on the off-policy Q-learning algorithm~\cite{sutton2018reinforcement}. 
Our proposed algorithm leverages the threshold-structure of the optimal policy developed in Section~\ref{sec:lagrangian-relaxation} while learning the Q-values for different state-value action pairs.

Since our problem is indexable, we focus on bandit $i$ for the ease of exposition.  The learning setup remains the same as in UCB-Whittle, i.e., the  algorithm utilizes an episodic framework in which each episode last for $H$ consecutive time steps. Recall the single-service MDP as in~\eqref{def:single_mdp1}:
\begin{align}\label{eq:obj-subproblem}
\min_{\pi_i}\limsup_{T\rightarrow\infty} \frac{1}{T}\mathbb{E}_{\pi_i}\Bigg\{\int_0^T\bigg( C_i(s,a)
 {-W}(1-a)           \bigg)  dt      \Bigg\},
\end{align}
where $s,a$ denote the state and action for a single bandit. The state-action value $\forall (s,a)$ under policy $\pi_i$ follows
\begin{align}\nonumber\label{eq:Q_value}
    Q^{\pi_i}(s,a):=a(C_i(s,1)+\sum_{s^\prime}p_i(s^\prime |s,1)V^{\pi_i}(s^\prime))-f_i\\
 +(1-a)(C_i(s,0)-W_i(s)+\sum_{s^\prime}p_i(s^\prime |s,0)V^{\pi_i}(s^\prime)),
\end{align}
where $V^{\pi_i}(s^\prime)$ denotes the cost-to-go function  for bandit $i$ in state $s^\prime$ under the policy $\pi_i$,
\begin{align} \nonumber\label{eq:dynamic_programming}
& V^{\pi_i}(s)=\underset{a\in\{0,1\}}{\min}{a(C_i(s,1)+\sum\limits_{s^\prime}p_i(s^\prime |s,1)V^{\pi_i}(s^\prime))-f_i}+\\
&(1-a)\times(C_i(s,0)-W_i(s)+\sum\limits_{s^\prime}p_i(s^\prime |s,0)V^{\pi_i}(s^\prime)), 
\end{align}
 and $f_i$ is the optimal average cost. 
Since the Whittle index in state $s$ is defined as the value $W_i(s)$ such that actions $a=0$ and $a=1$ are equally preferred in current state $s$, i.e., $Q^{\pi_i}(s,0)=Q^{\pi_i}(s,1)$, we  have 
\begin{align}\label{eq:Whittle index_formulation}\nonumber
    W_i(s)&=C_i(s,0)-C_i(s,1)\\
    &\quad\qquad+\sum_{s^\prime}(p_i(s^\prime|s,0)-p_i(s^\prime|s,1))V^{\pi_i}(s^\prime).
\end{align}
However, the unknown transition probabilities $p_i(s^\prime|s,a)$ hinders us to directly calculate Whittle indices according to \eqref{eq:Whittle index_formulation}.

A novel Q-learning based method has recently been proposed in \cite{avrachenkov2020whittle} to jointly update $Q^{\pi_i}(s,a)$ and $W_i(s)$ in a two time-scale stochastic process with a rigorous convergence analysis.  It leverages the $\epsilon$-greedy mechanism for balancing the tradeoff between exploration and exploitation according to \eqref{eq:Q_value} and \eqref{eq:Whittle index_formulation}.  More specifically, when bandit $i$ moves into a new state $s^\prime$, with probability $\epsilon$, it greedily selects an action $a^\prime$ to maximize $Q(s^\prime, a^\prime)$, and with probability $1-\epsilon$, it selects a random action.  As a result, the Whittle indices for all states are coupled and the Q values need to be updated for every state-action pair.

Given that the single-service problem in \eqref{eq:obj-subproblem} is indexable with an appealing threshold-type optimal policy, our key observation here is that the Whittle index $W_i(s)$ for state $s$ is only a function of Whittle indices for previous states, i.e., $\{W_i(r), \forall r<s\}$.  This is because under a threshold policy $\pi_i=R,$ any state $s<R$ is made passive and $s\geq R$ is made active otherwise.  As a result, for any state of the Q function, the Q-learning following the threshold policy $\pi_i=R$ will select a deterministic action as  $a=0, \forall s<R$ and $a=1,$ otherwise.  In other words, the Q-learning under threshold policy $\pi_i=R$ only updates $Q^R(s,0), \forall s<R$ and $Q^R(s,1), \forall s\geq R$, while keeping the rest state-action values unchanged.  Following the definition of Whittle index and Proposition \ref{prop:whittle-index-closed}, the desired Whittle index $W_i(s)$ satisfies 
$Q^{s}(s,1)=Q^{s+1}(s,0).$
Therefore, given the value function $V^{\pi_i}(s)$ under threshold policy $\pi_i=R$, it is clear that $Q^s(s,1)$ is only a function of $\{W_i(r), \forall r<s\}$ and $Q^{s+1}(s,0)$ is only a function of $\{W_i(r), \forall r<s+1\}$.

Upon this key observation and inspired by \cite{avrachenkov2020whittle}, we propose a heuristic state-by-state Q-learning based Whittle index policy, named Q-learning-Whittle.  In particular, 
the Q-learning recursion for threshold $\pi_i=R$ is defined as \eqref{eq:Q_update} for each episode $k$, 
\begin{figure*}
\begin{align}\label{eq:Q_update}
  Q_{t_k+1}^R(s,a)=
\begin{cases}
  (1-\alpha(t_k))Q_{t_k}^R(s,1)+\alpha(t_k)(C_i(s,1)
  +Q_{t_k}^R(s^\prime, \mathbb{1}_{s^\prime\geq R})), \quad\text{if $s\geq R$},\\
  (1-\alpha(t_k))Q_{t_k}^R(s,0)+\alpha(t_k)(C_i(s,0)-W_{i,k}(s)
  +Q_{t_k}^R(s^\prime, \mathbb{1}_{s^\prime\geq R})), \quad\text{if $s<R$},
\end{cases}
\end{align}
\vspace{-0.15in}
\end{figure*}
where $t_k\in\{\tau_k, \tau_k+1, \ldots, \tau_k+H\}$ with $\tau_k:=H(k-1)$ being the start time slot for episode $k$. {$W_{i,k}(s)$ is the Whittle index for episode $k$ and is kept fixed for the entire episode.}
The learning rate $\alpha(t_k)$ satisfies 
    $\sum_{t=0}^\infty \alpha(t)=\infty$ and $\sum_{t=0}^\infty \alpha(t)^2<\infty$ \cite{borkar2009stochastic}.

Our goal is to learn the Whittle index $W_i(s)$ by iteratively updating $W_{i,k}(s)$ in the following manner
\begin{align}\label{eq: W_update}
    W_{i, k+1}(s)&=(1-\gamma(k))W_{i, k}(s)\nonumber\\
    &\quad\qquad+\gamma(k)(Q_{\tau_{k+1}}^{s+1}(s,0)-Q_{\tau_{k+1}}^s(s,1)),
\end{align}
with $\gamma(h)$ satisfying
    $\sum_{k=0}^\infty \gamma(k)=\infty,$ and $\sum_{k=0}^\infty \gamma(k)^2<\infty.$

In summary, for each state $s$, the Whittle index $W_i(s)$ is updated according to the two time-scale stochastic process defined in \eqref{eq:Q_update} and \eqref{eq: W_update}. Note that the Whittle indices in Q-learning-Whittle are only be updated at the start of each episode according to \eqref{eq: W_update} and this reduces the two time-scale stochastic process into two independent one time-scale stochastic processes.  This property makes our Q-learning-Whittle substantially enhance the sample efficiency with a faster convergence speed, which will be verified in the numerical evaluations in Section~\ref{sec:sim}.  We are also hopeful that this nature can benefit the convergence analysis in \cite{avrachenkov2020whittle} and we leave this as part of our future work. 
The entire procedure  is summarized in Algorithm \ref{Algorithm2}.

\begin{algorithm}
\caption{Q-learning-Whittle}
\label{Algorithm2}
\begin{algorithmic}[1]
\State Initialize: $Q^{\pi_i}(s,a)=0,~ W_i(s)$ for all $\pi_i\in\mathcal{S}$.
\For {$s\in\mathcal{S}$}
\State Set the threshold policy as $\pi_i=s$.
\For{Episodes $k=1,2,\ldots,$}
\State The initial state is assumed to be $s=0$ for all bandits.
\State Update $W_{i,k}(s)$ according to \eqref{eq: W_update}.
\For{$h=1,2,\ldots,H$}
\State Update $Q_{\tau_k+h}^s(s,a)$ according to~\eqref{eq:Q_update}.
\State Update $Q_{\tau_k+h}^{s+1}(s,a)$ according to~\eqref{eq:Q_update}.
\EndFor
\EndFor
\State $k\leftarrow k+1$.
\EndFor
\State Return: $W_i(s), \forall s\in\mathcal{S}$.
\end{algorithmic}
\end{algorithm}

Note that Algorithm \ref{Algorithm2} is only for  a single bandit. The server just needs to repeat the procedures of Algorithm \ref{Algorithm2} for $N$ times to achieve the Whittle indices for all bandits. Since there is no difference in the update mechanism for different bandits,  we omit the algorithm description for the whole system.  When considering the capacity of the server, i.e., it can only serve a maximum number of $K$ bandits, an easy implementation of the algorithm is to divide the bandits into multiple groups with size $K$, and the server sequentially learns the Whittle indices for the bandits in each group.

%% file: proofucb.tex
\section{Proofs of Main Results}\label{sec:learning-proof}

In this section, we provide the proofs of Theorem~\ref{th:ucb_regret} on UCB-Whittle.  
For the ease of exposition, we relegate the proof details of the preliminary results below to Appendix~\ref{sec:appendix-ucb}.

Throughout, we will make the following assumption regarding the set $\Theta$ that is known to contain the true parameter $\te_0$.
\begin{assumption}\label{assum:1}
The process $\vec{S}(t)$ that denotes the {number of services} under the application of Whittle policy $W_{\te_0}$, is ergodic~\cite{meyn2012markov}, i.e., the Markov chain is positive recurrent. Moreover, the associated average cost $\bar{C}(W_{\te_0}; \te_0)$ is finite.
\end{assumption}

\subsection{Preliminary Results}

We provide the following equivalent formulation of UCB-Whittle. This characterization turns out to be useful while analyzing its learning regret.
\begin{lemma}\label{lemma:equiv}
UCB-Whittle can equivalently be described as follows. At the beginning of $\cE_k$, the learning algorithm computes a certain ``value'' $I_{\te}(\tau_k)$ for each $\te\in \Theta$ as follows,
\begin{align}\label{def:index}
    I_{\te}(\tau_k) := \min_{\te\up\in \cO(\tau_k)} \bar{C}(W_{\te};\te\up),
\end{align}
where $\cO(t)$ is the confidence interval at time $t$~\eqref{def:ci}. It then chooses $\tilde{\te}_k$ the $\te$ that has the least value $I_{\te}(\tau_k)$, and implements the corresponding Whittle index rule $W_{\tilde{\te}_k}$ during $\cE_k$.
\end{lemma}

We now prove a concentration result for $\hat{\te}(t)$.
\begin{lemma}\label{lemma:ci_fail}
Define the set (event) $\cG_1$ as follows,
\begin{align}\label{def:g1}
    \cG_1:= \left\{\omega \in \Omega: \te_0 \in \cO(t)~\forall t\in [1,T]  \right\}.
\end{align}
Consider UCB-Whittle employed with the parameter $K_1$ equal to $2\tau^2_h$, where the parameter $\tau_h$ is chosen so as to satisfy $\tau_h \ge \frac{-\log \epsilon}{\lambda}$. We have,
\begin{align}\label{prob_g1}
\bP( \cG_1 ) \ge 1 - (\delta +\epsilon).
\end{align}
\end{lemma}

We now show that with a high probability, the value $I_{\te_0}(t)$ of the Whittle policy corresponding to the true value $\te_0$ never falls above a certain threshold value of $\bar{C}(W_{\te_0};\te_0)$.
\begin{lemma}\label{lemma:true_lb}
On the set $\cG_1$ we have $I_{\te_0}(t) \le \bar{C}(W_{\te_0};\te_0),~ \forall t \in [1,T] $.
\end{lemma}

\begin{lemma}\label{lemma:ergodic}
Consider the process $\vec{S}(t)$ that evolves under the application of $W_{\tilde{\te}_k}$, where $\tilde{\te}_k$ is the optimistic estimate obtained at the beginning of $\cE_k$ by solving~\eqref{def:tilde_te}. Then, on the set $\cG_1$,
the following holds: there exists an $\eta>0$ such that for each {service $i$},
\begin{align*}
    \frac{1}{H}\bE \left( \sum_{t\in \cE_k}  A_i(s) \Big| \mathcal{F}_{\tau_k}\right)\ge \eta,
\end{align*}
where $A_i(s)$ is $1$ if {service $i$ is placed at time $s$}, and is $0$ otherwise. A similar inequality is also satisfied by the cumulative arrivals.
\end{lemma}
We next show that if {each service} has been sampled ``sufficiently'' many times, then the value $I_{\te}(t)$ attached to a sub-optimal $\te$ stays above the threshold $\bar{C}(W_{\te_0};\te_0)$.
Let
\begin{align}\label{def:k1}
k_1 := K_3 \log\left(\frac{N T^{b}}{\delta} \right),
\end{align}
where $K_3 \ge \frac{2 K_1}{\eta H \epsilon^2_1},$
and the parameter $\epsilon_1$ satisfies $\epsilon_1\le \epsilon(\Delta)$ (see Lemma~\ref{lemma:exp_concentration}). Define $\cG_2$ to be the following set (event)
\begin{align}\label{def:g2}
\cG_2 := \left\{\omega \in \Omega: N_i(\tau_{k_1}) \ge \frac{y_i}{2},\forall i \right\},
\end{align}
where in the above we denote
\begin{align*}
    y_i := \sum_{k=0}^{k_1} \bE\left( \sum_{t\in\cE_k} A_i(t)| \mathcal{F}_{\tau_k}\right).
\end{align*}
On the set $\cG_2$ we can ensure lower bounds on the number of samples $N_i(t)$ of a service, since they are tightly concentrated around its mean value. Since the regret at any time $t$ increases with the size of confidence balls~\eqref{def:ci}, and since this size decreases with $N_i(t)$, we are able to obtain tight upper-bounds on regret on $\cG_2$.

\begin{lemma}\label{lemma:3}
Let $k_1$ be as in~\eqref{def:k1}. {On the set $\cG_1 \cap \cG_2$} we have that for episodes $k>k_1$, the value $I_{\te}(t)$ of any sub-optimal $W_{\te}$ is greater than $\bar{C}(W_{\te_0};\te_0)$.
\end{lemma}

\subsection{Regret Analysis}
Now we are ready to prove the regret of UCB-Whittle.
We begin with the following result that allows us to decompose the regret into sum of ``episodic regrets''.
\begin{lemma}\label{lemma:decompose}
For a learning algorithm $\phi$, we can upper-bound its regret $R(\phi,T)$~\eqref{def:regret} by a sum of ``episodic regrets'' as follows,
\begin{align*}
    \bE \left( R(\phi,T) \right) \le \sum_{k=1}^{K} \bE \left[\id\left\{\tilde{\te}_k \neq \te_0 \right\} \max_{\te \in \Theta} \bar{C}(W_{\te};\te_0)\right],
\end{align*}
where $K$ denotes the number of episodes until $T$.
\end{lemma}
\begin{proof}
The proof follows since in each episode, the learning algorithm uses a stationary control policy from within the set of stabilizing controllers $\left\{ W_{\te}:\te\in\Theta \right\}$.
\end{proof}

\begin{proof}[\textbf{Proof of Theorem~\ref{th:ucb_regret}}]
Lemma~\ref{lemma:decompose}, the decomposition result shows us the following:
\begin{enumerate}
    \item Episodic regret is $0$ in those episodes in which $\tilde{\te}_k$ is equal to $\te_0$.
    \item If $\tilde{\te}_k$ is not equal to $\te_0$, then the regret is bounded by $\max_{\te\in\Theta} \bar{C}(W_{\te};\te_0)$ times $H$.
\end{enumerate}
Define $\cG:= \cG_1\cap \cG_2$ as the good set, where $\cG_1$ and $\cG_2$ are as in~\eqref{def:g1} and~\eqref{def:g2} respectively. Note that it follows from Lemma~\ref{lemma:ci_fail} and Lemma~\ref{lemma:g2} that the probability of $\cG$ is greater than $1-\left(\delta + \epsilon + \frac{\delta}{T}\right)$.

We thus have the following three sources of regret:

(i) \emph{Regret due to sub-optimal episodes on $\cG$}: It follows from Lemma~\ref{lemma:true_lb} and~\ref{lemma:3} that the regret is $0$ in an episode if $k>k_1$.

(ii) Regret on $\cG^{c}_1$: As is seen in the proof of Lemma~\ref{lemma:ci_fail}, the probability that the confidence interval $\cO(\tau_k)$ in episode $k$ fails is upper-bounded as $\frac{\delta}{T^b}$. Since the regret within such an episode is bounded by $H\max_{\te\in\Theta} \bar{C}(W_{\te};\te_0)$, the cumulative expected regret is upper-bounded by $ \frac{K H\max_{\te\in\Theta} \bar{C}(W_{\te};\te_0)}{T^{b}}= \frac{\max_{\te\in\Theta} \bar{C}(W_{\te};\te_0)}{T^{b-1}}$, where the constant $b>1$.

{(iii) Regret on $\cG^{c}_2$: Since the probability of this set is $O\left(\frac{1}{T}\right)$, and the regret on $\cG^{c}_2$ can be trivially upper-bounded by $T$, this component is $O(1)$.} 
The proof of the claimed upper-bound on expected regret then follows by adding the above three bounds.
\end{proof}

%% file: simulation.tex
\begin{table}
\centering
\begin{tabular}{|c|ccccccc|}
 \hline
 $\lambda/\mu$ & 1 & 2 & 3 & 4 & 5 & 6 & 7 \\
 \hline
Index & 4.46  &  3.35  &  3.11   & 1.06   & 1.231   & 0.706
   & 2.55   \\
   \hline
\end{tabular}
\caption{Relative optimality gap in $\%$. }
\label{table:gap}
\vspace{-0.2in}
\end{table}

\section{Simulation Results}\label{sec:sim}

In this section, we first validate the performance of our proposed Whittle index policy in Section~\ref{sec:lagrangian-relaxation}, and then evaluate the performance of UCB-Whittle and Q-learning-Whittle.

\subsection{Performance of Whittle Index Policy}
We consider two classes of services with $\mu_1=\mu_2=5$ and $\lambda_1=\lambda_2$.  The cost function is given by a linear function as $C_k(S_k,a)=S_k/\lambda_k$ in \eqref{eq:cost} for $k\in\{1,2\}$. We compare the relative error of the Whittle index policy with respect to the optimal value iteration policy, which optimalitily solves the  MDP problem in \eqref{eq:obj-continous-DP}. The relative optimality gap for different $\lambda/\mu$ is given in Table \ref{table:gap}.  We observe that our Whittle index policy achieves nearly optimal performance across all loads.

We also depict the actions taken under Whittle index policy and the optimal policy with $\lambda_1/\mu_1=4$ and $\lambda_2/\mu_2=6$ under the two bandits system.  Figure~\ref{fig:comparison} shows the switching curves for serving any one of the two bandits under these two policies.  The curves divide the entire space into two distinct regions, in which bandit $1$ and bandit $2$ are served  respectively.  It is clear that our Whittle index policy coincides with the optimal policy almost across all the state space, and hence captures the qualitative structure of the optimal policy.

\subsection{Performance of Learning Algorithms}

We consider $N$ services and at each time slot a maximum number $K=N/2$ services can be placed at the edge server. The bandits have a common set of state $\mathcal{S}=\{0, 1, 2, 3, 4, 5\}$, where $s\in\mathcal{S}$ represents the queue length of the service. We assume every bandit has the same delivery time and we set $\mu=5$. The horizon of each episode is set as $H=100$.

In Figure~\ref{Fig.convergence}, we verify the convergence of UCB-Whittle and Q-learning-Whittle and compare them with the traditional $\epsilon$-greedy Q-learning based index policy \cite{avrachenkov2020whittle}. We assume that all bandits share the same linear cost function as $C_i(S_i,a)=S_i/\lambda_i$, which represents the average serving time for bandit $i$ according to Little's law~\cite{john1961little,kleinrock1976queueing}.  For Q-learning-Whittle, the training parameters are set as $\alpha(t_k)=0.01,  \gamma(k)=0.005$.  From Figure~\ref{Fig.convergence}, we observe that the Whittle indices for each state obtained by both UCB-Whittle and Q-learning-Whittle converge as the number of episodes goes large.  Moreover, the UCB-Whittle algorithm has faster convergence speed for the Whittle indices of states with a longer queue length, whereas it requires more iterations to guarantee the convergence of those Whittle indices of states with a shorter queue length.  Specifically, the Whittle indices of Q-learning-Whittle converges state by state, and the whole learning process converges roughly in $180$ episodes, while UCB-Whittle converges roughly in $120$ episodes.  However, the conventional $\epsilon$-greedy Q-learning based indices, i.e., the benchmark in \cite{avrachenkov2020whittle}, requires a much larger number of episodes to converge.  This demonstrates the sample efficiency of our proposed algorithms.

\begin{figure*}
	\center
	\begin{minipage}[b]{.32\textwidth}
    	\includegraphics[width=0.95\textwidth]{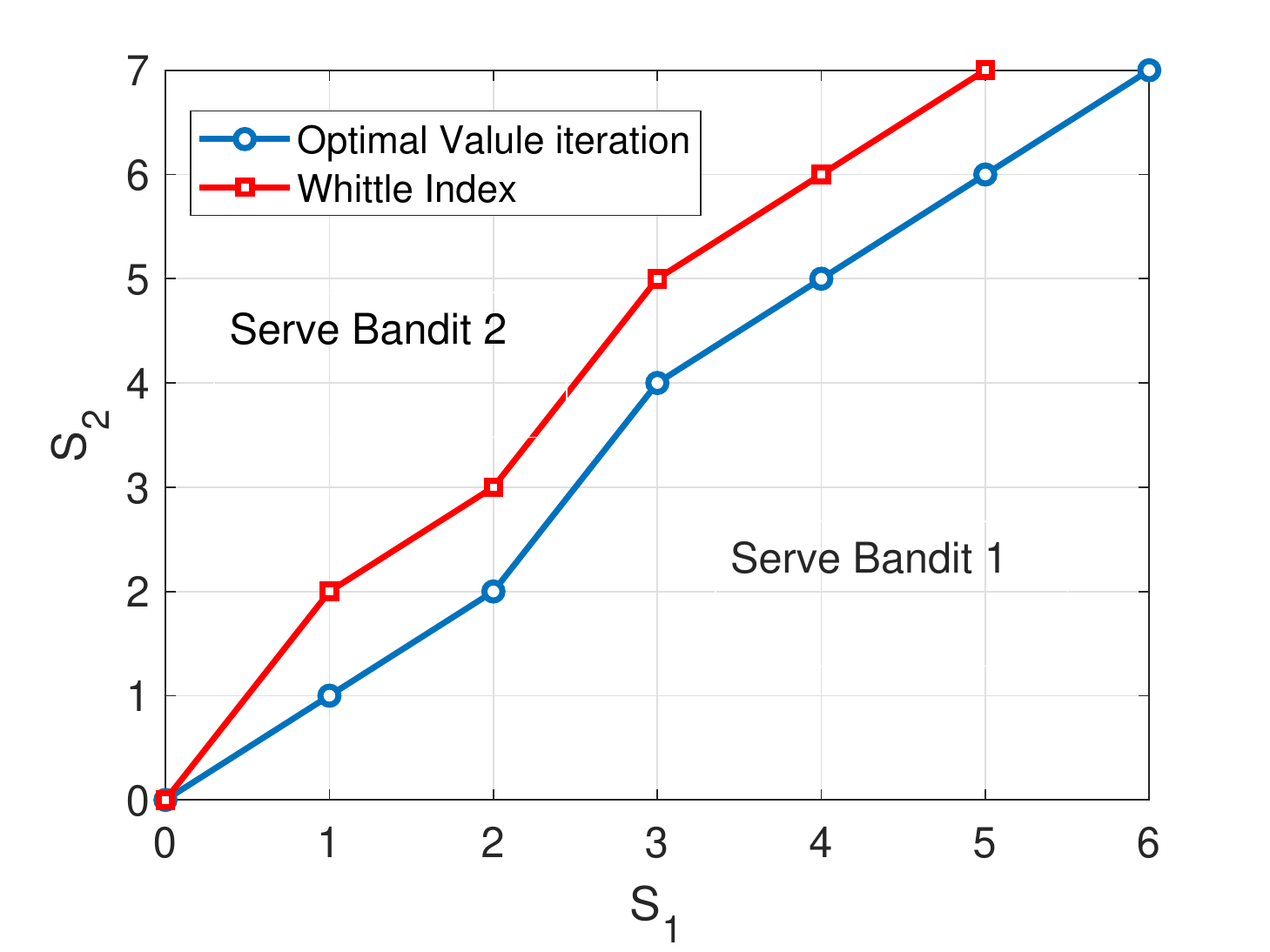}
\caption{Comparison between the optimal policy and Whittle index policy.}
\label{fig:comparison}
	\end{minipage}
	\begin{minipage}[b]{.32\textwidth}
	    \includegraphics[width=0.95\textwidth]{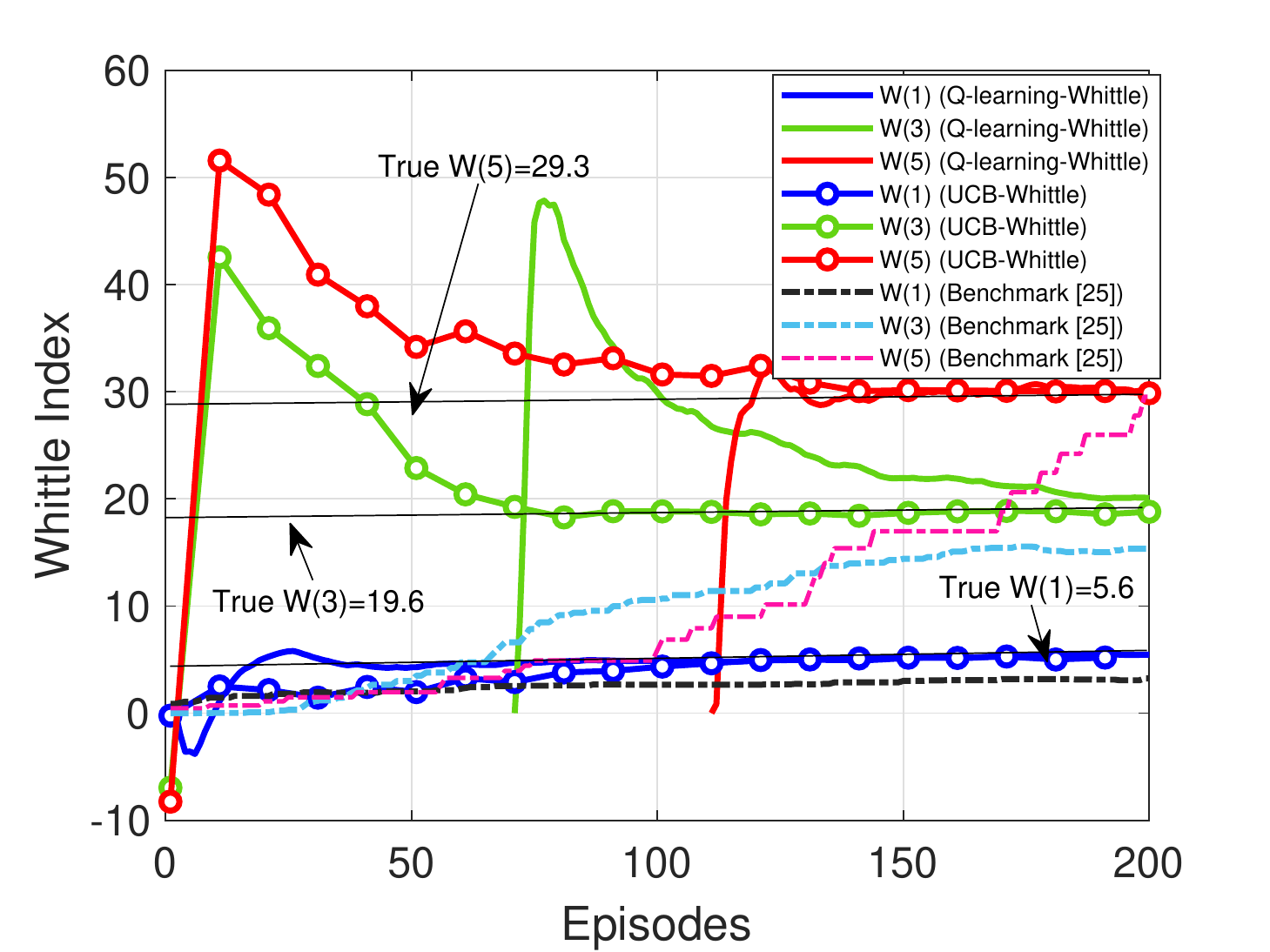}
\caption{Convergence of UCB-Whittle and Q-learning-Whittle. }
\label{Fig.convergence}
	\end{minipage}
	\begin{minipage}[b]{.32\textwidth}
	   \includegraphics[width=0.95\textwidth]{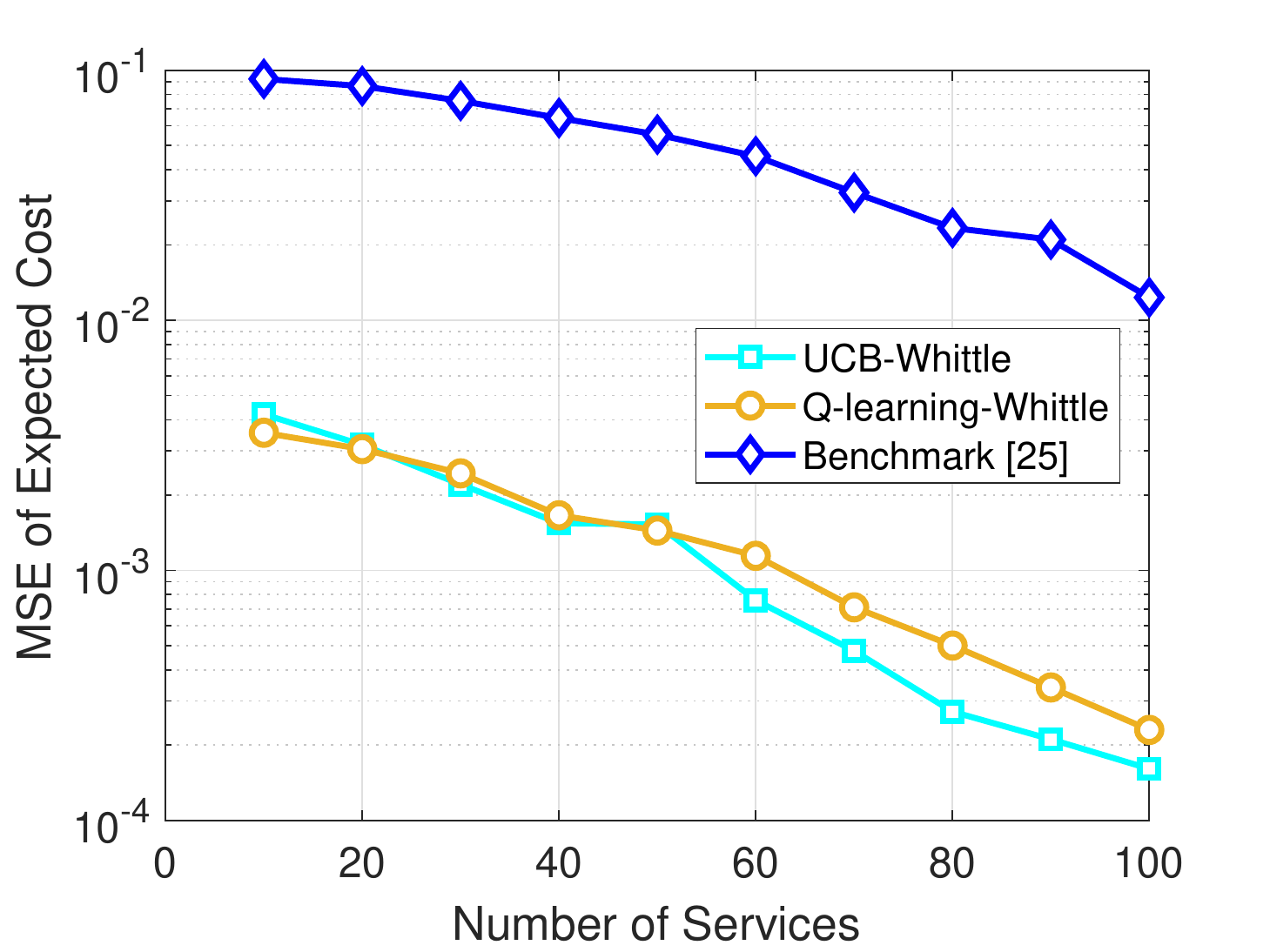}
\caption{MSE of expected cost comparison.} \label{Fig.cost}
	\end{minipage}
\end{figure*}

We further compute the mean squared error (MSE) of the expected costs achieved by UCB-Whittle and Q-learning-Whittle with that of the policy derived by the true Whittle index (with known system parameters) with the number of bandits varies from $N=10$ to $N=100$. We assume every bandit belong to one of $5$ different types with each type has equal number of bandits. 
The number of training episodes is set as $K=150$.
The cost function is $C_i(S_i,a)=S_i/\lambda_i$, where $\lambda_i\in\{10, 15,  20,  25,  30\}$.   {Since the learned Whittle indices converge to the true Whittle indices as the number of episodes goes large, }
the MSE of expected costs achieved by UCB-Whittle and Q-learning-Whittle converge. As shown in Figure \ref{Fig.cost}, the MSE of expected costs monotonically decrease as the number of services grows. This indicates the asymptotic optimality of the proposed algorithms.  Specifically, UCB-Whittle  and Q-learning-Whittle have similar performance when number of bandits is small, and the former slightly outperforms the latter for a larger number of bandits.  Nevertheless, both our proposed algorithms significantly outperform the conventional $\epsilon$-greedy Q-learning algorithm.  This also verifies the performance of our proposed algorithms.

%% file: related.tex
\section{Related Work}\label{sec:related}

\textbf{Service Placement.} The problem of service placement at the network edge has recently attracted a lot of attentions.  For example, approximation algorithms are introduced in \cite{he2018s,farhadi2019service,pasteris2019service} to find a feasible service placement to maximize the total user utility through the formulation of mixed integer linear programs.  Joint optimization of service placement and access point selection or request routing are addressed in \cite{gao2019winning,poularakis2019joint,lin2020service}. A unified service placement and request dispatching framework is proposed to optimize tradeoffs between latency and service cost in \cite{yang2015cost} while \cite{zeng2016joint} studies a similar problem in the context of image placement in fog network.  However, most of existing results are under the assumption of full knowledge of the entire resources and users or the knowledge of statistics are assumed to be known. Thus none of these works incorporate learning component.  The few works incorporate learning components on resource allocation and provisioning at cloud and edge computing are \cite{ouyang2019adaptive,chen2018spatio}. Our work differs significantly from these by posing the service placement problem at the network edge as a MDP under a continuous time model and developing an appealing simple threshold-type policy with performance guarantees that can be easily learned through a RL framework.

\textbf{Restless Multi-Armed Bandits (RMAB).}  RMAB has been applied in numerous applications for resource allocation, such as congestion control \cite{avrachenkov2013congestion}, cloud computing \cite{borkar2017index}, queueing systems \cite{archibald2009indexability}, wireless communication \cite{raghunathan2008index,liu2010indexability} and so on.  Most existing works on RMAB focus on offline setting, e.g., \cite{liu2010indexability,whittle1988restless} while one often cannot have full system information beforehand in practice.  Online RMAB has also gained attention, e.g., \cite{tekin2012online,jung2019regret,liu2011logarithmic}. While its wide applications, RMAB suffers the curse of dimensionality and is provably hard (i.e., PSPACE \cite{papadimitriou1994complexity}). One very successful heuristic has been the celebrated Whittle index policy \cite{whittle1988restless}, where the hard constraint of the number of bandits that can be used at each time is relaxed to be on average. Many studies focus on finding Whittle index policy for RMAB, e.g., \cite{kadota2018scheduling,kadota2016minimizing,guo2016index,singh2015index}, however, all are under specific assumptions and hard to generalize. Furthermore, the application of Whittle index requires full knowledge of the system, which is often not the case in practice. Many existing algorithms such as the classical adaptive control schemes and RL schemes do not exploit the special structure available in the problem. 

Recent works such as ~\cite{borkar2018reinforcement,avrachenkov2019learning,fu2019towards} have developed learning algorithms that utilize Q-learning and multi-time scale stochastic approximation algorithms. However, they lack finite-time guarantees on ``learning regret''~\cite{auer2002finite,lattimore2020bandit}. Multi time-scale stochastic approximation algorithms typically suffer from the problem of slow convergence. A naive application of existing RL algorithms such as UCRL2, Thompson Sample or RBMLE~\cite{jaksch2010near,ouyang2017learning,mete2020reward} to solve our problem, would yield a learning regret that scales linearly with the state-space size. Since the size of state-space for our problem grows exponentially with the number of services $N$ that could be placed at the server, this would mean that the performance of these algorithms is absymal, and consequently they are of limited practical value.

%% file: conclusion.tex
\section{Conclusion}\label{sec:conclusion}

In this paper, we studied the problem of optimal service placement at the network edge with the goal to minimize average service delivery latency.  The problem formulation took the form on a Markov decision process.  To overcome the computational complexity of solving the corresponding MDP, we showed that the optimal policy for the MDP associated with each service is of threshold-type and derived explicitly the Whittle indices for each service. We then developed efficient learning augmented algorithms, UCB-Whittle and Q-learning-Whittle to utilize the structure of optimal policies given that service request and delivery rates are usually unknown.  We characterized the performance of these two algorithms and also numerically demonstrated an excellent empirical performance.  

%% file: appendixnonlearning.tex
\section{Appendix}\label{sec:app}

\subsection{Proofs on Whittle Index Policy in Section~\ref{sec:lagrangian-relaxation}}\label{sec:app-non-learning}
In this subsection, we present the proofs details of the threshold policy (Proposition~\ref{prop:threshold-policy}), stationary distribution of the threshold policy (Proposition~\ref{prop:stationary-distributions}), the indexability property of all bandits in our model (Proposition~\ref{prop:indexable}) as well as the computation of Whittle indices (Proposition~\ref{prop:whittle-index-closed}).

\subsubsection{Proof of Proposition~\ref{prop:threshold-policy}}
\begin{proof}
From our model, we have $b_i(S_i, a_i)=\lambda_i(S_i)$ and $d_i(S_i, a_i)=\mu_i(S_i)A_i$ where $A_i=1$ if service $i$ is placed on the edge server.  
Since the set of feasible policies $\Pi$ is non-empty, there exists a stationary policy $\pi^*$ that optimally solves~(\ref{def:single_mdp1}).
Let $R^*=\max\{S\in\{0, 1,\cdots\}: A_{\pi^*}(S)=0\}$, where we use the subscript $\pi^*$ to denote actions under policy $\pi^*$. Based on the definition of transition rates, we have $d_i(S_i,0)=0$ for all states $s_i<R^*$. Zero departure rate indicates that all states below $R^*$  are transient, implying the stationary probability for bandit $i$ in state $S_i$ under policy $\pi^*$ being $0$, i.e., $q_{\pi^*, i}(S_i)=0$, $\forall S_i<R^*$.  By definition, we have $A_{\pi^*, i}=0, \forall S_i\leq R^*$ and $A_{\pi^*, i}=1$, $\forall S_i>R^*.$ The average cost given by~(\ref{def:single_mdp1}) under the optimal policy $\pi^*$ then reduces to
\begin{align}
&\mathbb{E}_{\pi^*}[C_i(S_i, A_{\pi^*, i})] -W \mathbb{E}1_{\{ A_{\pi^*, i}=0\}}\nonumber\displaybreak[0]\\
=&C_i(R^*, 0) q_{\pi^*, i}(R^*)+\!\!\!\!\sum_{S_i=R^*+1}^{\infty}\!\! \! C_i(S_i, 1)q_{\pi^*, i}(S_i) - \!Wq_{\pi^*, i}(R^*)\nonumber\displaybreak[1]\\
=&\mathbb{E}_{R^*}[C_i(S_i, A_{R^*, i})] - W q_{R^*, i}(R^*),
\end{align}
indicating the threshold policy with threshold $R^*$ gives optimal performance.  This completes the proof.
\end{proof}

\subsubsection{Proof of Proposition~\ref{prop:stationary-distributions}}
\begin{proof}
For the ease of exposition, we only consider the state of service $i.$  Denote its queue length as $S_i.$
From our definition, the transition rate satisfies $\mathbb{P}(S_i, S_i+1)=\lambda_i$, and $\mathbb{P}(S_i, S_i-1)=0$ for $S_i\leq R$ and $\mathbb{P}(S_i, S_i-1)=\mu_i$ for $S_i> R$.  It is clear that the dummy states in which $R^\prime<R$ is transient because the queue length keeps increasing. Therefore, the stationary probabilities for dummy states are all zero, i.e, $q_i^R(R^\prime)=0.$  Based on standard birth-and-death process, the stationary probabilities of service $i$ can be expressed as
\begin{equation}
q_i(R+l)=\left(\frac{\lambda_i}{\mu_i}\right)^l\frac{1}{\Pi_{k=1}^l(R+k)} q_i(R).
\end{equation}
Since $q_i(R)+q_i(R+1)+q_i(R+2)+\cdots=1$, we obtain
\begin{align}
q_i(R)=1/\left(1+\sum\limits_{j=1}^{\infty}\left(\frac{\lambda_i}{\mu_i}\right)^j\frac{1}{\Pi_{k=1}^j(R+k)}\right).
\end{align}
This completes the proof.
\end{proof}

\subsubsection{Proof of Proposition~\ref{prop:indexable}}
\begin{proof}
Since the optimal policy for~(\ref{eq:obj-continous-relaxed}) is a threshold policy, for a given subsidy $W$, the optimal average cost under threshold $R$ will be $h(W):=\min_R\{h^R(W)\}$, where
\begin{align}\label{eq:threshold_cost}
h^R(W):=\sum_{s=0}^{\infty}C_i(s, A_{R, i})q_{R, i}(s)-W\sum_{s=0}^R q_{R, i}(s).
\end{align}
It is easy to show that $h(W)$ is a concave non-increasing function since it is a lower envelope of linear non-increasing functions in $W$, i.e., $h^R(W)<h^R(W^\prime)$ if $W<W^\prime.$  This means we can choose a larger threshold $R$ when $W$ increases to further decrease the total cost according to \eqref{eq:threshold_cost}, i.e, $R(W)\subseteq R(W^\prime)$ when $W<W^\prime$.

Next we show that $\sum_{s=0}^R q_i^R(s)$ is strictly increasing in $R$.  From Proposition~\ref{prop:stationary-distributions}, we have
\begin{align}
\sum_{s=0}^R q_i^R(s)=q_i^{R}(R^\prime)=\frac{1}{1+\sum_{j=1}^{\infty}\left(\frac{\lambda}{\mu}\right)^{j}\prod_{k=0}^{j-1} \frac{1}{R+1+k}},
\end{align}
which is strictly increasing in $R.$
\end{proof}

\subsubsection{Proof of Proposition~\ref{prop:whittle-index-closed}}
\begin{proof}
Consider the single-agent MDP for service $i$, that is parametrized by $W$. Note that from Proposition~\ref{prop:threshold-policy} we have that for this MDP, a threshold policy is optimal. It follows from the definition of Whittle index, that the performance of the policy with threshold at $R$ should be equal to the performance of a policy with threshold at $R+1$, i.e.,
\begin{align}
&\mathbb{E}_R[C_i(S_i, A_i(S_i))] -W \mathbb{E}_{R}1_{\{ A_i(S_i)=0\}} \nonumber\\
=&\mathbb{E}_{R+1}[C_i(S_i, A_i(S_i))] -W \mathbb{E}_{R+1}1_{\{ A_i(S_i)=0\}},
\end{align}
where a sub-script denotes the fact that the associated quantities involve a threshold policy with the value of threshold set equal to this value. Since the evolution of the $i$-th bandit is described by a  birth-and-death process, we have $\mathbb{E}_R 1_{\{ A_i(S_i)=0\}}=\sum_{S_i=0}^R q_{R, i}(S_i)$, where $q_{R, i}(S_i)$ is the stationary probability for bandit $i$ in state $s$ under the threshold policy $R$ (Proposition~\ref{prop:stationary-distributions}). This completes the proof.
\end{proof}

%% file: appendix.tex
\subsection{Proofs on UCB-Whittle}\label{sec:appendix-ucb}
In this subsection, we provide the proofs of Lemmas~\ref{lemma:equiv}-\ref{lemma:3} in Section~\ref{sec:learning-proof}.

\subsubsection{Proof of Lemma~\ref{lemma:equiv}}
\begin{proof}
Note that the operation of obtaining $\tilde{\te}_k$ could equivalently be re-written as follows
\begin{align}\label{eq:lemma:equiv1}
\min_{\te\in \Theta} \min_{\te\up\in \cO(t)} \bar{C}(W_{\te};\te\up).
\end{align}
After exchanging the order of minimization,~(\ref{eq:lemma:equiv1}) reduces to
\begin{align}\label{eq:lemma:equiv2}
\min_{\te\up\in \cO(t)} \min_{\te\in \Theta} \bar{C}(W_{\te};\te\up).
\end{align}
Since the Whittle index policy is asymptotically optimal~\cite{verloop2016asymptotically} and there are only finitely many candidate policies $W_{\te}$ in the inner-minimization problem above,~(\ref{eq:lemma:equiv2}) reduces to 
\begin{align}\label{eq:lemma:equiv3}
\min_{\te\up\in \cO(t)}  \bar{C}(W_{\te\up};\te\up).
\end{align}
However, this is exactly the problem~\eqref{def:tilde_te} that needs to be solved in order to obtain $\tilde{\te}_k$.
\end{proof}

\subsubsection{Proof of Lemma~\ref{lemma:ci_fail}}
\begin{proof}
Fix the number of samples used for estimating $m_i$ at $n_i$. It then follows from Lemma~\ref{lemma:exp_concentration} with the parameter $t$ set equal to $\sqrt{K_1 n_i \log\left(\frac{N\cdot T^{b}}{\delta} \right)}$ that the probability with which the estimate of service rate lies outside the confidence ball, is less than $  \exp\left( -\frac{K_1}{2\tau^{2}_h} \log\left(\frac{N\cdot T^{b}}{\delta} \right) \right)+\epsilon$. Letting {$K_1 = 2\tau^2_h$}, we have that this probability is upper-bounded by $\frac{\delta}{NT^b}$.

Since $n_i$ can possibly assume $T$ number of values, and {$b>1$}, the proof follows by using union bound on the number of estimation samples, and users.
\end{proof}

\subsubsection{Proof of Lemma~\ref{lemma:true_lb}}
\begin{proof} It follows from the equivalent definition~\eqref{def:index} of UCB-Whittle developed in Lemma~\ref{lemma:equiv} that we have $I_{\te_0}(t)\le \bar{C}(W_{\te_0};\te_0)$ if $\te_0\in\cO(t)$. The proof then follows from Lemma~\ref{lemma:ci_fail}.
\end{proof}

\subsubsection{Proof of Lemma~\ref{lemma:ergodic}}
\begin{proof}
Note that in order for the average cost to be finite, each service $i$ should be allocated a non-zero fraction of bandwidth at the server. Since under Assumption~\ref{assum:1} $W_{\te}$ the controlled process $\vec{S}(t)$ is ergodic and has finite cost, on $\cG_1$ the cost incurred under $W_{\tilde{\te}_k}$ is definitely finite, and hence each service is provided a non-zero fraction of the total bandwidth at the server. The proof then follows since there are only finitely many choices for $W_{\tilde{\te}_k}$ .
\end{proof}

\subsubsection{Proof of Lemma~\ref{lemma:3}}
\begin{proof}
On $\cG_2$ we have that the number of samples of $i$ satisfy $N_i(\tau_{k_1})\ge \frac{\eta H}{2} K_3 \log\left(\frac{N T^{b}}{\delta} \right)$, and hence it follows from the definition of confidence intervals~\eqref{def:radius} that on $\cG_1\cap \cG_2$ we have $\|\te_0 -\hat{\te}(\tau_k)\|_{\infty} \le \sqrt{\frac{2 K_1}{\eta H K_3}}\le \epsilon(\Delta)$. Let $\te^{\star}$ denote the solution of the problem~\eqref{def:index}.
It follows from Lemma~\ref{lemma:wc} that the following holds on $\cG$,
\begin{align*}
    |\bar{C}(W_{\te}; \te^{\star})-\bar{C}(W_{\te}; \te_0)|&\le \Delta,
    \mbox{ or } |I_{\te}(t) -\bar{C}(W_{\te}; \te_0)|\le \Delta.
\end{align*}
This completes the proof.  
\end{proof}

\subsubsection{Additional Results}
In this subsection, we provide the proofs of three lemmas that are needed for the proof of Theorem~\ref{th:ucb_regret}.

\begin{lemma}\label{lemma:wc}
For every $\delta_1>0$, there exists an $\epsilon(\delta_1)>0$ such that
\begin{align*}
    |\bar{C}(W_{\tilde{\te}}; \te\up)-\bar{C}(W_{\tilde{\te}}; \te)|\le \delta_1,~\forall \tilde{\te} \in \Theta,
\end{align*}
for all $\te\up$ satisfying  $\|\te\up-\te\|_{\infty}\le \epsilon(\delta_1)$.
\end{lemma}

\begin{proof} Fix the control policy at $W_{\tilde{\te}}$, and consider a sequence $\te^{(n)}, n\in \bN$ of parameters satisfying $\te^{(n)}\to \te$. It follows from Theorem 1 of~\cite{karr1975weak} that the sequence of stationary probability measures corresponding to the system in which the parameter is $\te^{(n)}$, and the controller is $W_{\tilde{\te}}$, converges to the measure corresponding to using policy $W_{\tilde{\te}}$ on the system with parameter $\te$. The proof is then completed by constructing $\te^{(n)}$ from elements of $\Theta$, and also observing that we have finitely many possibilities for $\tilde{\te}$.\end{proof}

Let $k_1 = O(\log T)$. Recall that $\cG_2$ was defined as follows,
\begin{align*}
\cG_2 = \left\{\omega \in \Omega: N_i(\tau_{k_1}) \ge \frac{y_i}{2},\forall i \right\},
\end{align*}
where in the above
\begin{align*}
    y_i = \sum_{k=0}^{k_1} \bE\left( \sum_{t\in\cE_k} \id\{ u(t)=i\}| \mathcal{F}_{\tau_k}\right).
\end{align*}
\begin{lemma}\label{lemma:g2}
We have
\begin{align*}
  \bP\left( \cG_2 \right) \ge 1- \frac{1}{T}.
\end{align*}
\end{lemma}
\begin{proof}
Define
\begin{align*}
    m_k :&= \left(\sum_{t\in\cE_k} \id\{ u(t)=i\}\right) -   \bE\left( \sum_{t\in\cE_k} \id\{ u(t)=i\}| \mathcal{F}_{\tau_k}\right),\\
    k &= 1,2,\ldots.
\end{align*}
Clearly, $\{m_k\}$ is a martingale difference sequence that is uniformly bounded by the duration of a single episode $H$.
Recall~\eqref{def:k1} that $k_1 = K_3 \log\left(\frac{N T^{b}}{\delta} \right)$. Let us fix the number of episodes at $k_1$. We have the following from the Azuma-Hoeffding's inequality~\cite{azuma1967weighted,lattimore2020bandit},
\begin{align*}
    \bP\left( \Big | \sum_{k=0}^{k_1} m_k  \Big|\ge  \epsilon \right) \le \exp\left(- \frac{\epsilon^2}{2H^2 k_1} \right).
\end{align*}
Since each time a service is placed, we obtain a sample for estimating the corresponding rate of download, we obtain,
\begin{align*}
    &\bP\left( \Big | N\lmu(\tau_{k_1})- \sum_{k=0}^{k_1}\bE\left( \sum_{t\in\cE_k} \id\{ u(t)=i\}\right)  \Big|\ge  \epsilon \right)\\
    &\le \exp\left(- \frac{\epsilon^2}{2H^2 k_1} \right),
\end{align*}
where $\epsilon>0$. Upon letting $\epsilon= \frac{k_1 \eta H}{2}$, and observing that from Lemma~\ref{lemma:ergodic} we have that $\sum_{k=0}^{k_1}\bE\left( \sum_{t\in\cE_k} \id\{ u(t)=i\}\right) \ge \eta H k_1$, the above inequality reduces to
\begin{align*}
    \bP\left(  N\lmu(\tau_{k_1})\le  \eta H k_1 \slash 2 \right)
    &\le \exp\left(- \frac{K_3 \eta^2}{8} \log\left(\frac{N T^{b}}{\delta} \right) \right)\displaybreak[0]\\
    &\le \frac{\delta}{N T^b}
    < \frac{\delta}{N T},
\end{align*}
where the second inequality follows since {$K_3 \ge \frac{8}{\eta^2}$}, and the last follows since $b>1$. Proof is then completed by using union bound on users.
\end{proof}

The following result is taken from~\cite{dubhashi2009concentration}.
\begin{lemma}\label{lemma:exp_concentration}
Let $x_1,x_2,\ldots,x_n$ be i.i.d. exponentially distributed random variables with mean $1\slash \lambda$. Let $\tau_h$ be a threshold that satisfies the following
\begin{align}
\exp(-\lambda \tau_h) \le \epsilon,
\text{or } \tau_h \ge \frac{-\log \epsilon}{\lambda}.\label{cond:threshold_cond}
\end{align}
Then
\begin{align*}
\bP\left( \sum_{\ell=1}^{n} x_{\ell} - n\slash \lambda > t \right) \le \exp\left( -\frac{t^2}{2n\tau^{2}_h}  \right)+\epsilon ,
\end{align*}
or equivalently
\begin{align*}
\bP\left( \sum_{\ell=1}^{n} x_{\ell} \slash n - 1\slash \lambda > t\slash n \right) \le \exp\left( -\frac{t^2}{2n\tau^{2}_h}  \right)+\epsilon.
\end{align*}
\end{lemma}